\documentclass[submission,copyright,creativecommons]{eptcs}
\providecommand{\event}{Post Proceedings of ThEdu'21}
\usepackage{breakurl}
\usepackage{hyperref}

\usepackage{xurl}
\usepackage{graphicx}
\usepackage{amsmath}
\usepackage{amsthm}
\usepackage{MnSymbol}

 \newtheorem{theorem}{Theorem}
 \newtheorem{definition}{Definition}

\title{Automated Discovery of Geometrical Theorems in GeoGebra}
\author{Zoltán Kovács
\institute{The Private University College of Education of the Diocese of Linz\\Linz, Austria}\thanks{The
work was partially supported by the grant 
PID2020-113192GB-I00 from the Spanish MICINN.}
\email{zoltan@geogebra.org}
\and
Jonathan H.~Yu
\institute{Gilman School\\Baltimore, Maryland, USA}
\email{jonathanhy314@gmail.com}
}
\def\titlerunning{Automated Discovery in GeoGebra}
\def\authorrunning{Kovács \& Yu}
\begin{document}
\maketitle

\begin{abstract}
We describe a prototype of a new experimental GeoGebra command and tool, \texttt{Discover},
that analyzes geometric figures for salient patterns, properties, and theorems.
This tool is a basic implementation of automated discovery in elementary planar geometry.
The paper focuses on the mathematical background of the implementation,
as well as methods to avoid combinatorial explosion when storing
the interesting properties of a geometric figure.
\end{abstract}

\bibliographystyle{eptcs}

\section{Introduction}

In this paper we introduce a new GeoGebra command and tool \texttt{Discover},
which is available in the development GitHub repository\footnote{\url{https://github.com/kovzol/geogebra-discovery}}.
This research is closely related to the former
project \textit{Automated Geometer}\footnote{\url{https://github.com/kovzol/ag}}
(see \cite{adg-ag,aisc-ag,LNAI11110-ag} for further details).

Given a Euclidean geometry construction drawn in GeoGebra, suppose a user
wants to know if a given object $O$ has some ``interesting features,''
such as relevant theorems or properties.
This object can be a point, a line, a circle, or something else, although
in the current implementation $O$ will always be a point
contained in a given construction.
Without any further user input,
the \texttt{Discover} command will then analyze $O$
for its interesting and relevant features, and present them to the user
as a list of formulas, as well as graphical illustrations.

\begin{figure}\centering
\includegraphics[scale=0.75]{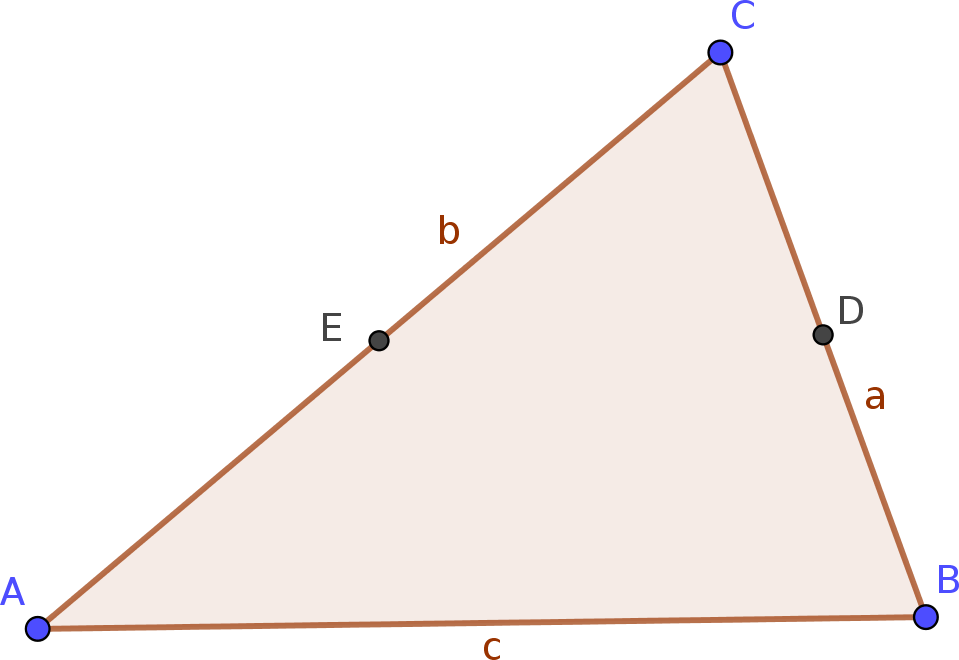}
\caption{Initial setup for a discovery}
\label{midline1}
\end{figure}

For example, let $ABC$ an arbitrary triangle, and let $D$ and $E$ be
the midpoints of $BC$ and $AC$, respectively (Fig.~\ref{midline1}). Has $D$ any
noteworthy features?
Yes: $DE$ is parallel to $AB$, independent of the position of $A$, $B$ and $C$.
Indeed, the command \texttt{Discover($D$)} confirms this observation with
the output shown in Fig.~\ref{rel-midline}; GeoGebra adds lines $DE$ and $AB$
in the same color (Fig.~\ref{midline2}).
(Note, however, that the current implementation of GeoGebra does not report that $2\cdot|DE|=|AB|$.)
Also, the software reports the somewhat trivial finding that the segments $BD$ and $CD$
are congruent, with $BD$ and $CD$ highlighted in the same color.
%(Actually, $AE$ and $CE$ are also congruent, but they do not have a direct relationship
%with the point $D$.)
(The fact that ``$AE$ and $CE$ are congruent'' is not reported because these segments do not have
a direct relationship with the point $D$.)

This output can be
obtained by selecting the Discover tool 
\raisebox{-.15\height}{\includegraphics[width=0.4cm]{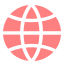}}
in GeoGebra's toolbox,
and then clicking on the point $D$. %This functionality is implemented in GeoGebra versions Classic 5 and 6,
%available as an experimental software package called \textit{GeoGebra Discovery}%
%\footnote{\url{http://github.com/kovzol/geogebra-discovery}}.% Incorporating of our improvements
%into the official GeoGebra versions is already scheduled by The GeoGebra Team.
This functionality is currently implemented in both GeoGebra versions Classic 5 and 6
and will soon be incorporated into the official GeoGebra versions.

\begin{figure}\centering
\includegraphics[scale=0.4]{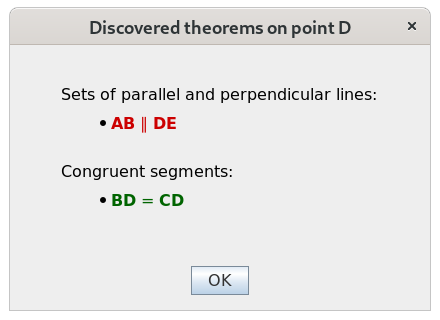}
\caption{Output window of the \texttt{Discover} command that reports the \textit{Midline theorem}}
\label{rel-midline}
\end{figure}

\begin{figure}\centering
\includegraphics[scale=0.75]{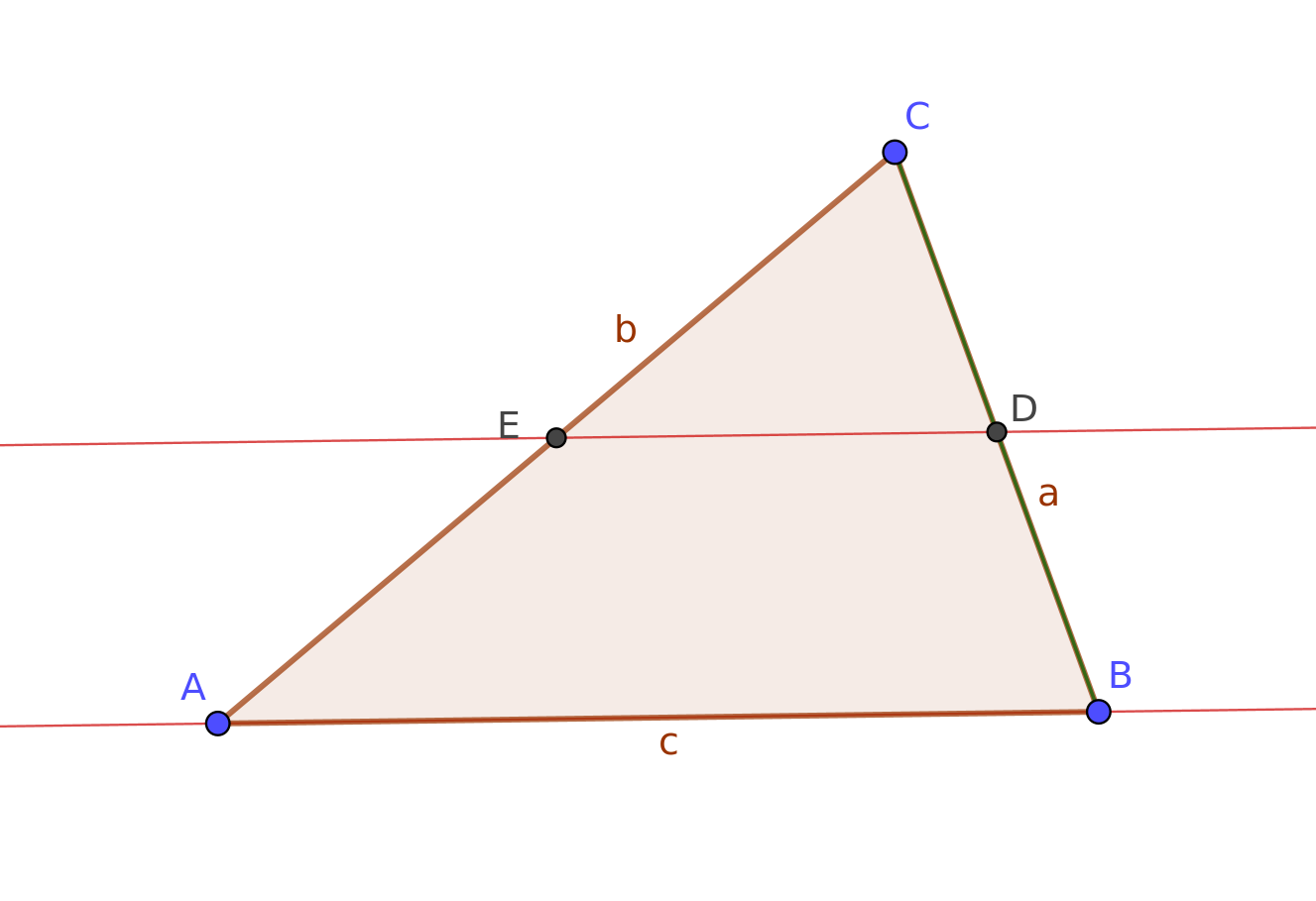}
\caption{Further output of the \texttt{Discover} command}
\label{midline2}
\end{figure}

The \texttt{Discover} command calculates its outputs with the following algorithm:
First, all points are analyzed to determine whether they are
the same as another point.
Second, all possible point triplets are examined for collinearity.
Third, all possible subsets containing four points on the figure are checked
for concyclicity.
With knowledge of collinear points, separate lines can be uniquely defined, in order
to determine whether they are parallel. 
Next, congruent segments can be identified by considering the pairs of all possible point pairs. 
Finally, perpendicular lines are identified.
This combination of numerical and symbolic processes provides the results of the \texttt{Discover} command.

Our second example illustrates a more complicated problem.
A regular hexagon  $ABCDEF$ is given in Fig.~\ref{hexagon1}.
Point $G$ is defined as the intersection of $AD$ and $BE$.
In addition, $H=BE\cap CF$, $I=AD\cap CF$.
The points $G$, $H$ and $I$ may have trivial differences in their numerical representations, but in the geometrical sense they should be equal.
For example, the $y$-coordinates of $G$ and $H$, as computed numerically by GeoGebra,
differ only in the least significant decimal place shown.
However, the final calculations to prove that they are identical will be symbolic and exact.

\begin{figure}\centering
\includegraphics[scale=0.4]{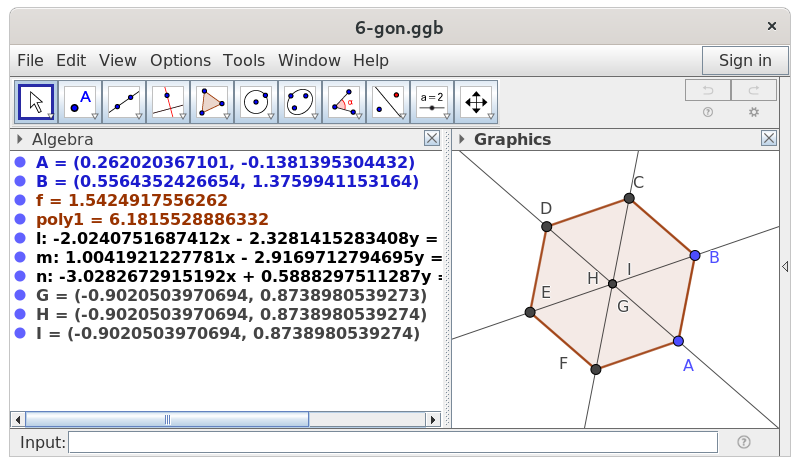}
\caption{Initial setup for another discovery}
\label{hexagon1}
\end{figure}

To determine any interesting features of point $F$,
we enter the command \texttt{Discover($F$)}.
GeoGebra reports a set of properties in a message box
and adds some additional outputs to the initial setup (Fig.~\ref{rel-hexagon}).

\begin{figure}\centering
\includegraphics[scale=0.4]{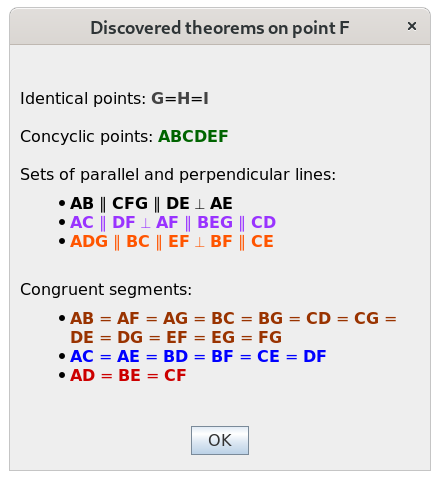}\ 
\includegraphics[scale=0.75]{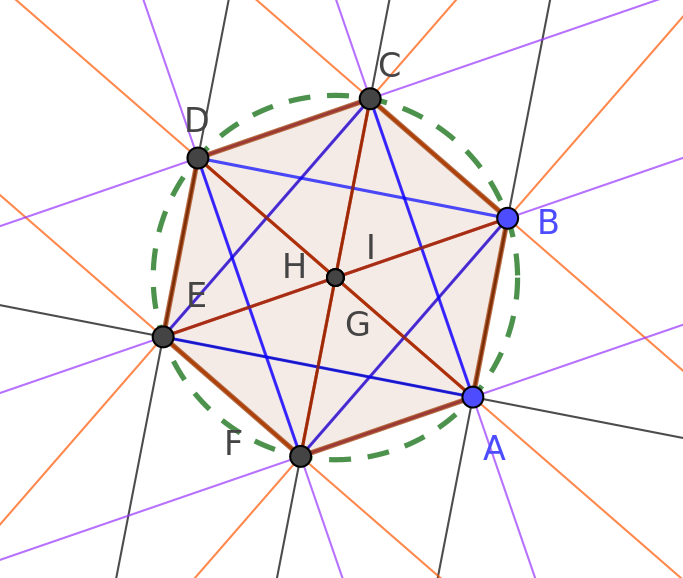}
\caption{Output window that reports several theorems related to point $F$,
and the geometric output of discovery}
\label{rel-hexagon}
\end{figure}

%Here we learn that the points $G$, $H$ and $I$ are all equal.
Here we learn that the points $G$, $H$ and $I$ are all equal, as stated above.
We also see that concylic points are reported as a single item, not as separate data.
In addition, parallel line sets and their perpendicular supplementors are classified into three different sets,
each colored with the same color.
This allows the user to distinguish among the rectangular grids related to the objects of the figure.
Finally, there are three sets of congruent segments.
This approach in computation and reporting helps avoid combinatorial explosion.

\section{Mathematical background}\label{sec2}

The above mentioned strategies have some similarities to the ones introduced in \cite{song}, but here we
focus on minimizing the number of objects that have to be compared in the process
that practically compares all objects with all other objects.

Our current implementation deals with \textit{points}, \textit{lines},
\textit{circles} and \textit{parallel lines} (or \textit{directions}) and
their perpendicular supplementors, and \textit{congruent segments}.

A \textit{geometric point} $P$ is a GeoGebra object, described by the \texttt{GeoPoint} class%
\footnote{See GeoGebra's source code at \url{github.com/geogebra/geogebra} for more details.}. While we will not provide
a detailed definition of a geometric point, generally speaking, it is an object with a very complex structure
containing two real coordinates, several style settings (including size and color, for example)
and other technical details that are used in the application. Some geometric points are
dependent of other geometric points or other geometric objects---this hierarchy is stored in the set of \texttt{GeoPoint}s, too.

Independent of the detailed definition of a geometric point,
we can still define the notion of \textit{point} in our context.

\begin{definition} A set of geometric points ${\cal P}=\{P_1,P_2,\ldots,P_n\}$ is called a point if for all
different $P,Q\in {\cal P}$ the points $P$ and $Q$ are identical in general.
\end{definition}

Henceforth, unless otherwise mentioned, we will consider points according to the definition above,
not as geometric points.

Here, we do not precisely define when two points are identical \textit{in general}.
Instead, we will illustrate the concept of point identicality with the following example.
Consider geometric points $P_1$, $P_2$, $P_3$ and $P_4$ that form a parallelogram.
Now define $P_5$ and $P_6$ as the midpoint of $P_1$ and $P_3$, and $P_2$ and $P_4$,
respectively. With these definitions, $P_5$ and $P_6$ are identical, because the diagonals of a parallelogram
always bisect each other. In a dynamic geometry setting like GeoGebra, this simply means that by changing some points
of the set $\{P_1,P_2,P_3,P_4\}$, the points $P_5$ and $P_6$ will still share the same position in the plane. (See
Fig.~\ref{parallelogram1}. Here the construction is determined only by the points $P_1$, $P_2$ and $P_3$:
after they are freely chosen, the point $P_4$ must be dependent and uniquely defined as
the intersection of the two parallel lines to $P_1P_2$ and $P_2P_3$, respectively, through $P_3$ and $P_1$.)

\begin{figure}\centering
\includegraphics[scale=0.75]{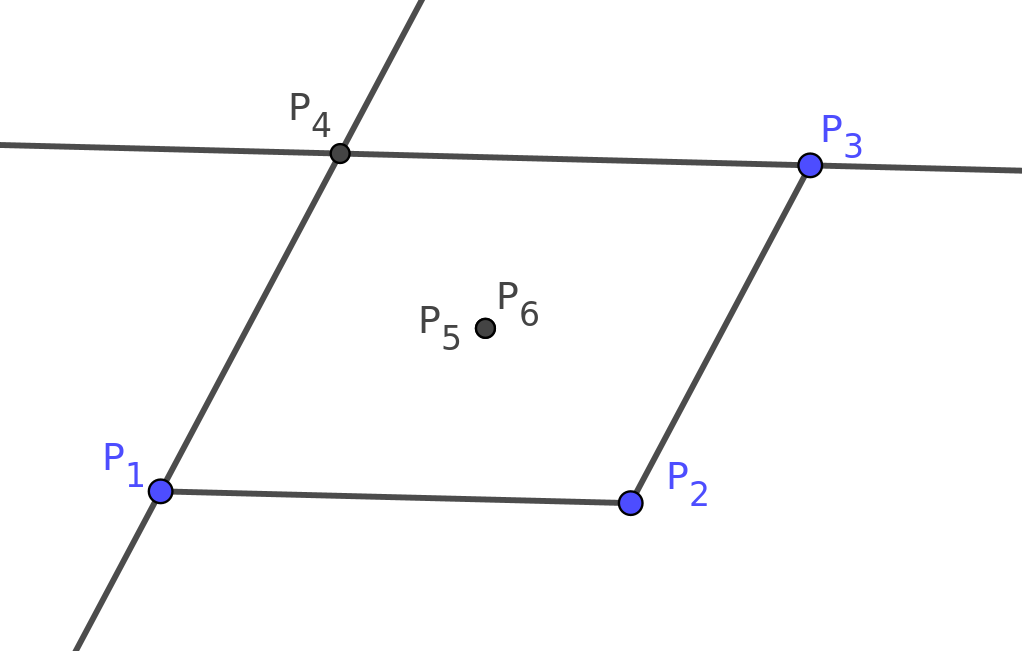}
\caption{Points $P_5$ and $P_6$ are defined as midpoints of opposite vertices of parallelogram $P_1P_2P_3P_4$}
\label{parallelogram1}
\end{figure}

Statements considered ``generally true'' are true in most typical cases,
but not 100\% of the time, especially in cases of degenerate objects.
For example, it is generally true that altitudes of a triangle meet at a point---but not always, since a degenerate triangle
``usually'' has three parallel ``altitudes,'' unless two (or even three!) vertices of the triangle coincide.
(See \cite{Chou_1987} for more details on the concept of general truth and degeneracies.)

\begin{definition} A set of points $\ell=\{P_1,P_2,\ldots,P_n\}$ is called a line if for all
different $P,Q,R\in \ell$ the points $P$, $Q$ and $R$ are collinear in general.
\end{definition}

For example, the set $\ell=\{C,F,G\}$ in Fig.~\ref{rel-hexagon} forms a line.

\begin{definition} A set of points ${\cal C}=\{P_1,P_2,P_3,\ldots,P_n\}$ is called a circle if for all
different $P,Q,R,S\in {\cal C}$ the points $P$, $Q$, $R$ and $S$ are concyclic in general.
\end{definition}

\begin{definition} A set of lines $\vec{D}=\{\ell_1,\ldots,\ell_n\}$ is called parallel lines (or a direction) if for all
different $\ell,m\in \vec{D}$ the lines $\ell$ and $m$ are parallel in general.
\end{definition}

\begin{definition} A pair of parallel lines $\vec{D_1}$, $\vec{D_2}$ is called perpendicular or parallel lines if for all
different $\ell\in \vec{D_1}$, $m\in \vec{D_2}$ the lines $\ell$ and $m$ are parallel or perpendicular in general.
We will use the notation $\vec{D_1}\upVdash\vec{D_2}$ in this case.
\end{definition}

\begin{definition} A set $\overline{s}=\{P,Q\}$ of two points is called a segment.
\end{definition}

\begin{definition} A set of segments $s=\{\overline{s_1},\ldots,\overline{s_n}\}$ is called equal length segments (or congruent segments) if for all
different $\overline{s_1},\overline{s_2}\in s$ the segments $\overline{s_1}$ and $\overline{s_1}$ are equally long in general.
\end{definition}

In fact, GeoGebra Discovery uses
a more general concept of being identical: it allows two points (or two objects) to have a kind of
relationship also if it is true just \textit{on parts} (see \cite{rmc-top} for more details).

The main idea of storing the objects is that points, lines, circles, directions (and their perpendicular supplementors) and equally long segments designate equivalence classes, that is:

\begin{theorem}Let $\ell$ and $m$ be lines. Then, for all different points $P,Q,R\in m$,
if $\{P, Q\}\subset \ell$, then $R\in \ell$; that is, $\ell = m$.
\end{theorem}

\begin{proof}
In Euclidean geometry\footnote{The statement also holds in absolute geometry.}
two points always designate a unique line.
\end{proof}

\begin{theorem}Let ${\cal C}$ and ${\cal D}$ be circles. Then, for all different points 
$P,Q,R,S\in{\cal D}$, if $\{P,Q,R\}\subset {\cal C}$, then $S\in{\cal C}$; that is, ${\cal C}={\cal D}$.
\end{theorem}
\begin{proof}
In Euclidean geometry three non-collinear points always designate a unique circle.
\end{proof}

\begin{theorem}Let $\vec{D}$ and $\vec{E}$ be directions. Let $\ell\in \vec{D}$ and $m\in\vec{E}$. If $\ell\parallel m$
in general, then $\vec{D}=\vec{E}$.
\end{theorem}
\begin{proof}
This follows immediately from the transitive property of parallelism.
\end{proof}

\begin{theorem}The relation $\upVdash$ is an equivalence relation on the parallel lines.
\end{theorem}
\begin{proof}
Reflexivity and symmetry are obvious. To check transitivity we assume that $\vec{D_1}\upVdash\vec{D_2}$
and $\vec{D_2}\upVdash\vec{D_3}$ hold and verify these four possible cases:
\begin{center}
\renewcommand\arraystretch{1.3}
\setlength\doublerulesep{0.3pt}
\begin{tabular}{c||c|c}
& $\vec{D_2}\parallel \vec{D_3}$ & $\vec{D_2}\perp \vec{D_3}$ \\
\hline\hline
$\vec{D_1}\parallel \vec{D_2}$ & $\vec{D_1}\parallel \vec{D_3}$ & $\vec{D_1}\perp \vec{D_3}$\\
\hline
$\vec{D_1}\perp \vec{D_2}$ & $\vec{D_1}\perp \vec{D_3}$ & $\vec{D_1}\parallel \vec{D_3}$
\end{tabular}
\end{center}
We note that it is a requirement that all the lines are defined in the \textit{plane}.
\end{proof}

\begin{theorem}Let $s$ and $t$ be segments. Let $\overline{u}\in s$ and $\overline{v}\in t$. If $|\overline{u}|=|\overline{v}|$
in general, then $s_1=s_2$.
\end{theorem}
\begin{proof}
This is an immediate consequence of the transitive property of equality of lengths.
\end{proof}

By using these theorems we can maintain a minimal set of objects during discovery. See \cite{kovacs-yu-2020}
for a more detailed description of the technical internals.

\section{Examples of Discover with selected theorems}

GeoGebra is a well-known and widely used software tool in education, with meaningful potential for using
geometric discovery and exploration to teach elementary geometry. Even so, the range of mathematical knowledge
is broad, including secondary school topics, international math competitions, and higher level mathematics.
Below we examine selected theorems confirmed in the current implementation of Discover.

\subsection{Euler line}

The Euler line is a
line determined from any triangle that is not regular. It
passes through the orthocenter, the circumcenter and the centroid.
The problem is shown in Fig.~\ref{EulerLine1}.
\begin{figure}\centering
\includegraphics[scale=0.75]{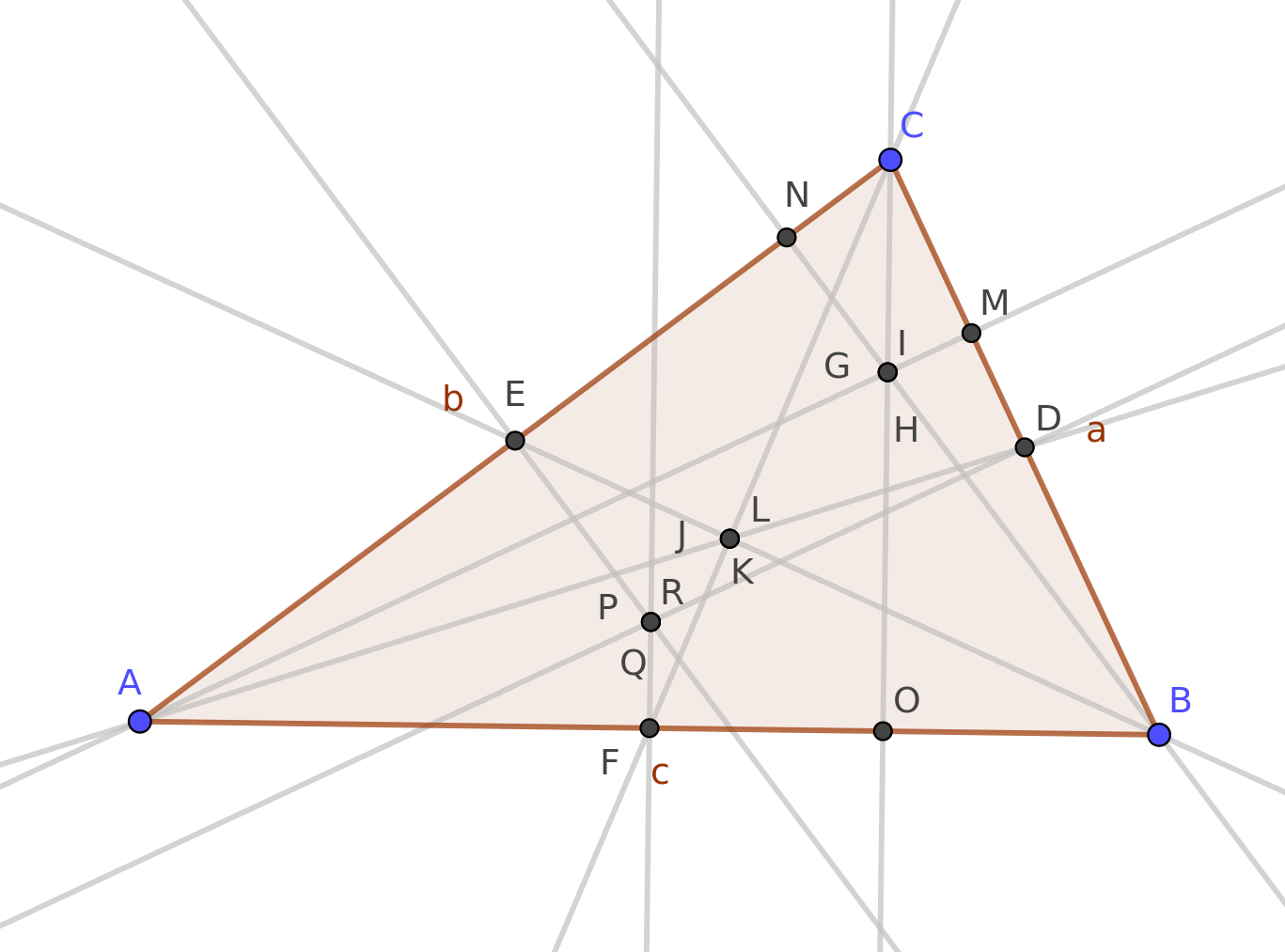}
\caption{Euler line}
\label{EulerLine1}
\end{figure}
With discovery on point $P$, the relevant theorems are listed in Fig.~\ref{EulerLine2}.
\begin{figure}\centering
\includegraphics[scale=0.65]{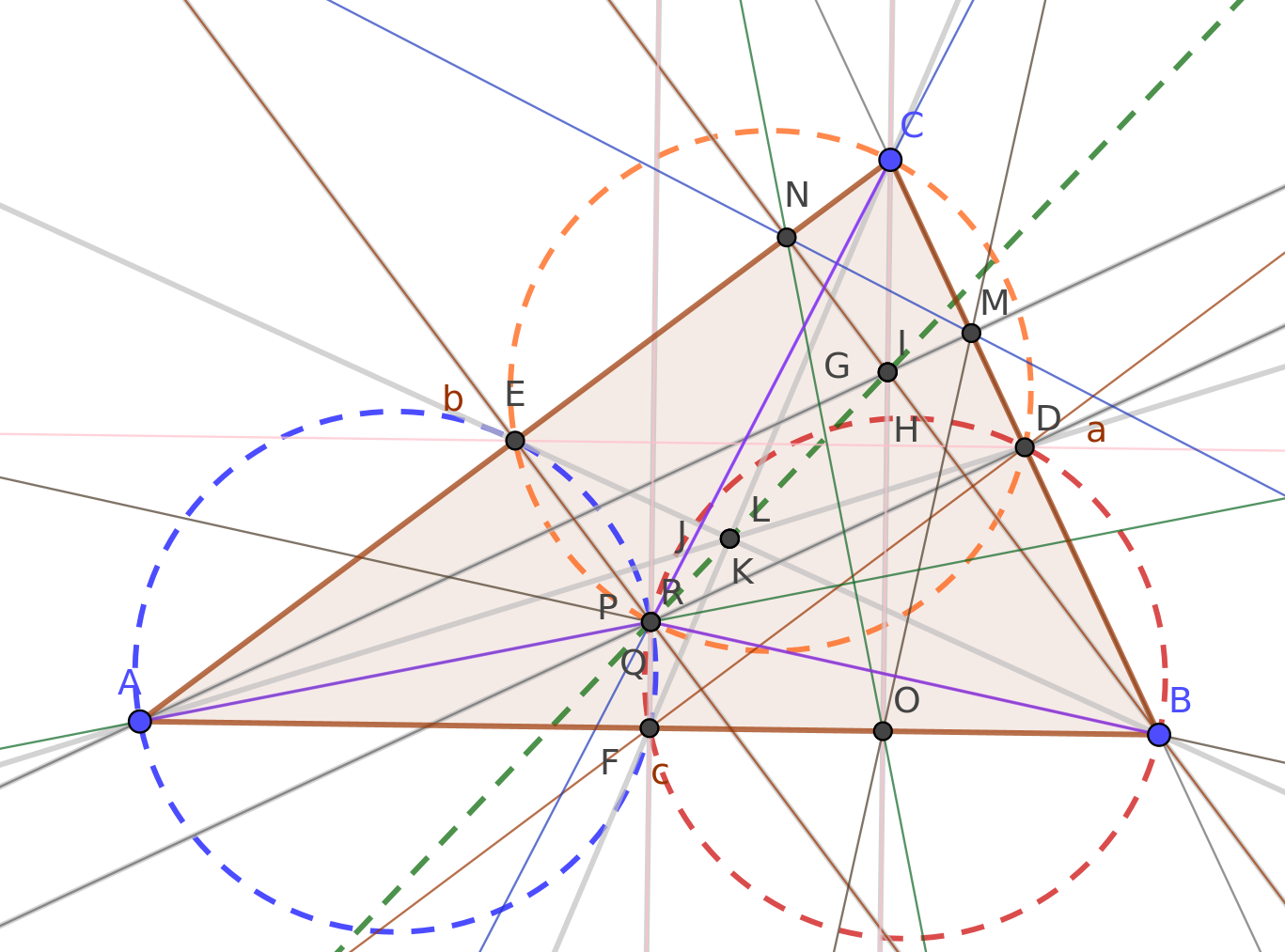}
\includegraphics[scale=0.35]{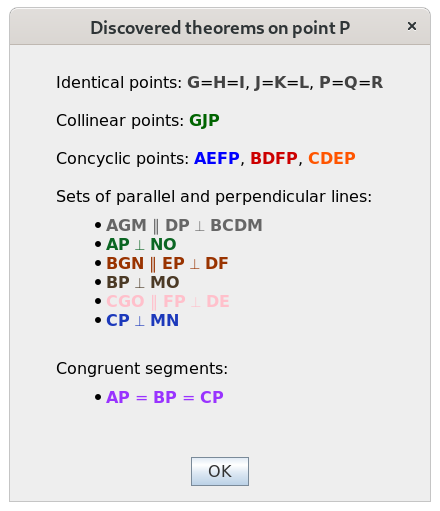}
\caption{Output of the command \texttt{Discover($P$)}}
\label{EulerLine2}
\end{figure}
The Euler line theorem implicitly includes several simple theorems, including 
concurrency of the altitudes ($G=H=I$, these points being the pairwise intersections of the altitudes),
concurrency of the medians
of a triangle ($J=K=L$, the generated points being the pairwise intersections of the medians),
and concurrency of the perpendicular bisectors of the altitudes ($P=Q=R$, pairwise intersections as above).

\subsection{Nine-point circle}

The nine-point circle passes through nine significant
points of an arbitrary triangle, namely:
\begin{itemize}
\item the midpoint of each side of the triangle,
\item the foot point of each altitude,
\item the midpoint of the line segment from each vertex of the triangle to the orthocenter.
\end{itemize}
With discovery on the midpoint $D$ of side $BC$, the appropriate theorems are reported in Fig.~\ref{9pc2}.
\begin{figure}\centering
\includegraphics[scale=0.75]{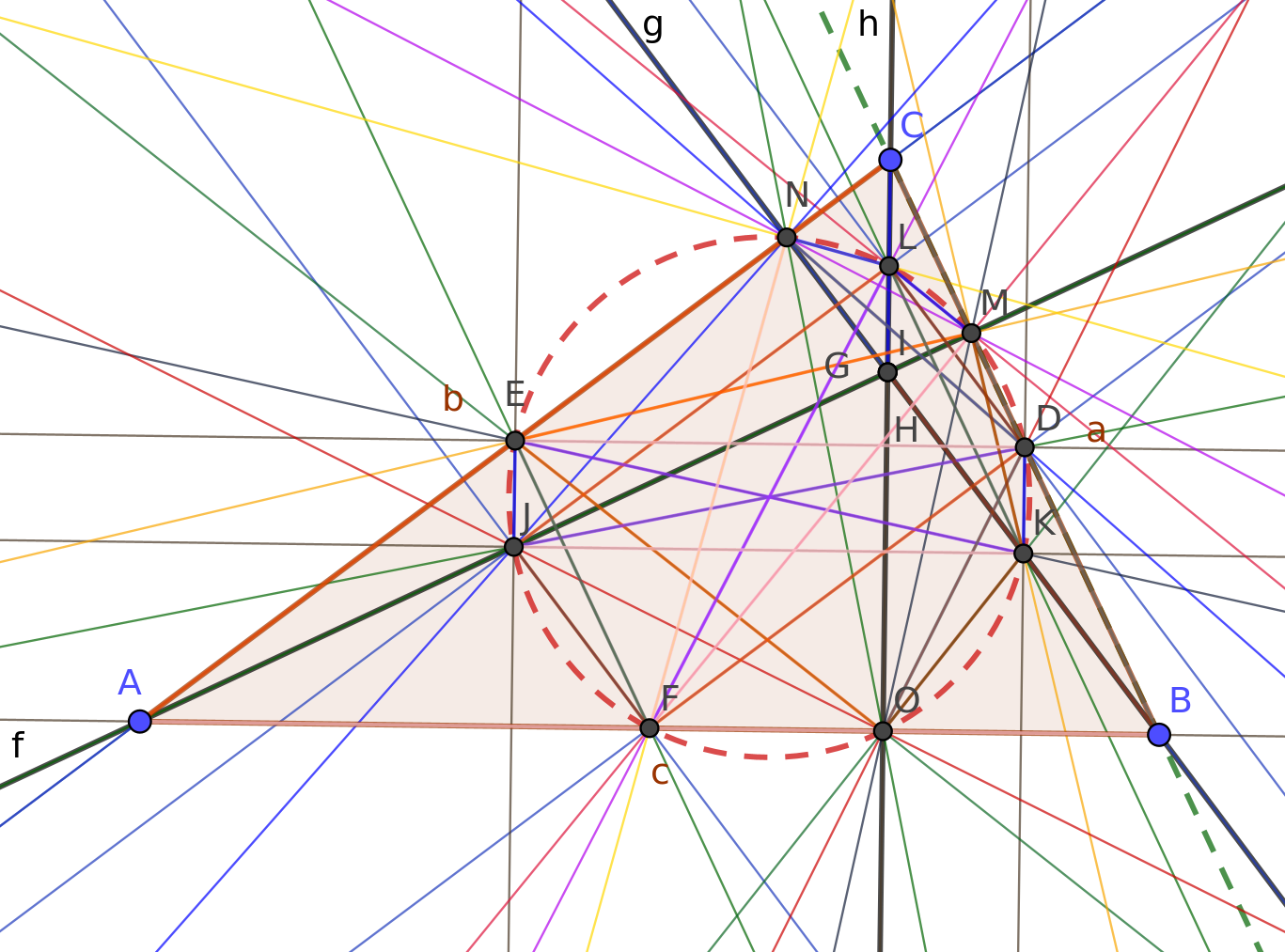}
\includegraphics[scale=0.4]{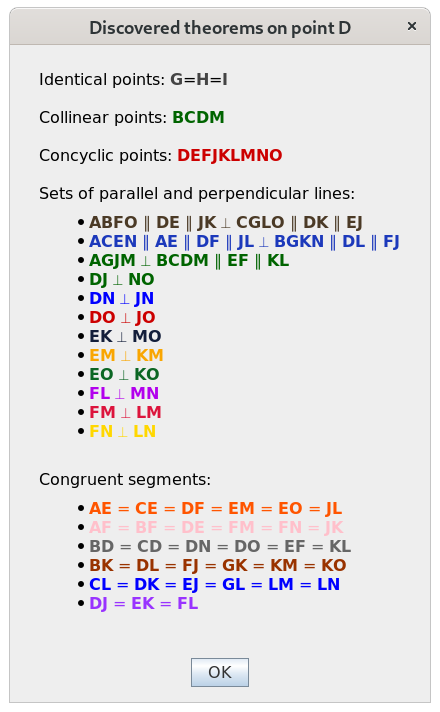}
\caption{Output of the command \texttt{Discover($D$)}}
\label{9pc2}
\end{figure}
The nine-point circle theorem implicitly includes several other simple theorems. In addition, the graphical result
suggests further theorems: segments $DJ$, $EK$, and $FL$ are congruent and concurrent;
these three segments are also diameters of the nine-point circle; and
their intersection designates the center of the nine-point circle. 
This can be easily confirmed by performing another discovery.

\subsection{A contest problem}

In 2010, at the 51st International Mathematics Olympiad in Astana, Kazakhstan, the following shortlisted
problem was proposed by United Kingdom:
\begin{quotation}
Let $ABC$ be an acute triangle with $D$, $E$, $F$ the feet of the altitudes lying on $BC$, $CA$, $AB$
respectively. One of the intersection points of the line $EF$ and the circumcircle is $P$. The lines
$BP$ and $DF$ meet at point $Q$. Prove that $AP=AQ$.
\end{quotation}
After constructing an acute triangle with GeoGebra Discovery (noting that
no matching figure can be drawn for a non-acute input angle), we start
discovery on point $Q$.

The discovered theorems appear in Fig.~\ref{contest2}.
We learn a few unexpected properties: $DP\parallel EQ$, and the points $C$, $D$, $P$, $Q$,
and $A$, $F$, $P$, $Q$ are concyclic.
\begin{figure}\centering
\includegraphics[scale=0.66]{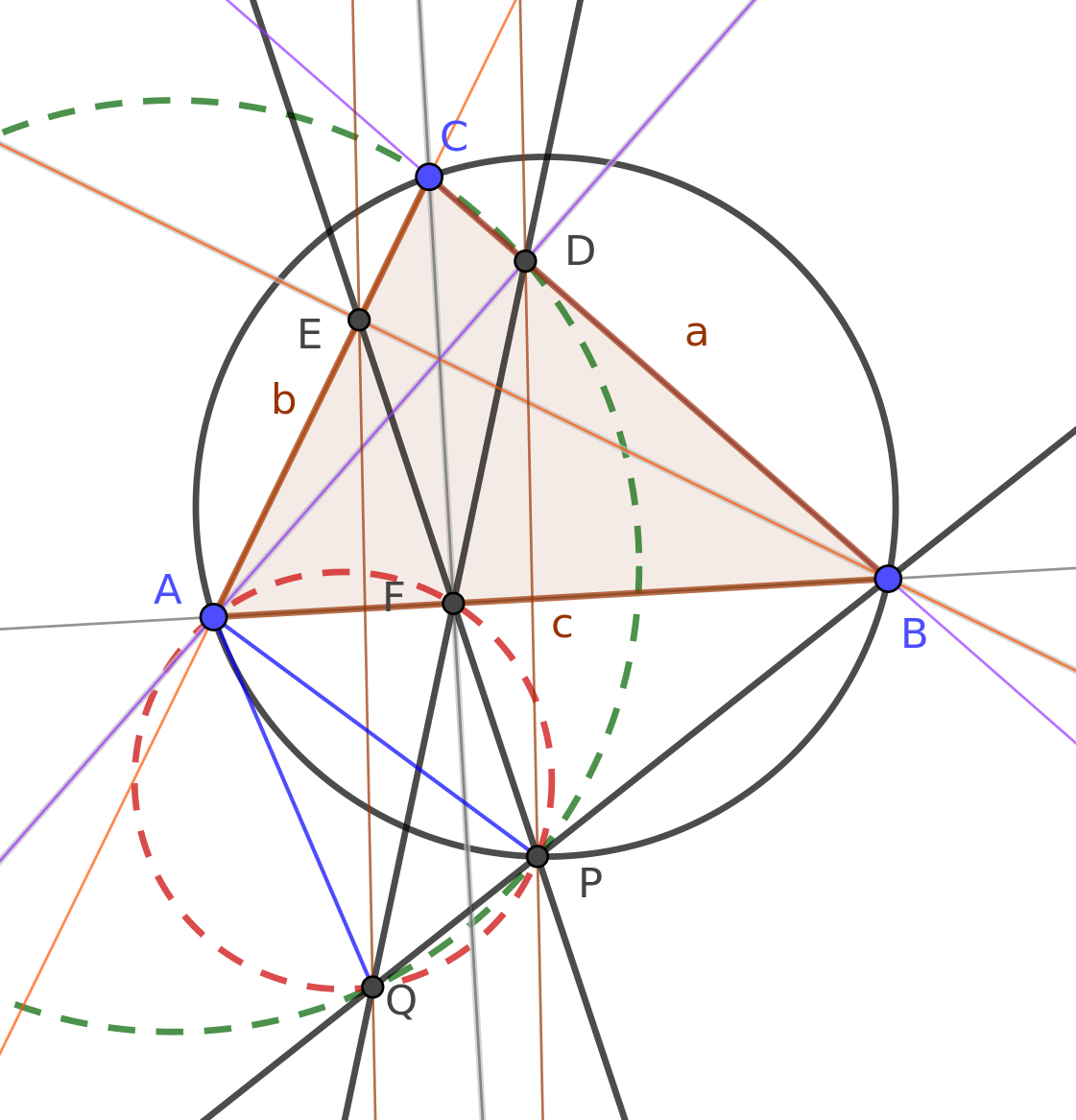}
\includegraphics[scale=0.4]{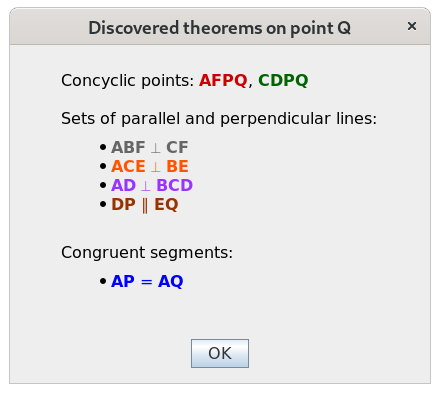}
\caption{Output of the command \texttt{Discover($Q$)}}
\label{contest2}
\end{figure}

\section{Discussion}\label{sec4}

\subsection{Trivial statements and theorems}

After the Discover command determines the salient properties for a given input,
it displays relevant \textit{theorems} but does not report \textit{trivial statements}.
One example of a trivial statement occurs in Fig.~\ref{rel-midline}, in which the\
collinearity of points $B$, $C$ and $D$ and of points $A$, $C$ and $E$ are not reported.
By defining $D$ as the midpoint of $BC$, this collinearity is implicitly assumed,
so it does not make any sense to reiterate this.

The question of which properties are considered trivial or not
has been a topic of long term discussion.
Over 30 years ago, in his famous book, Larry Wos asked,
``What properties can be identified to permit an automated reasoning
program to find new and interesting theorems, as opposed to proving conjectured theorems?''
 \cite{Wos88}.
This decision, at some level, does become a judgment call that may vary depending on the expertise of the user.
Most users would regard the above statement $BD=CD$ as trivial, given that $D$ is the midpoint of $BC$. On the other hand, explicitly stating this property might be helpful
for some beginners.
Currently, GeoGebra Discovery regards as trivial
all statements that directly follow
from the parameters of the original problem, 
as demonstrated in the example above.
However, more detailed criteria to determine which theorems
to display could be clarified in the future.
Recent research may help in crafting effective algorithms that prioritize identifying
the most interesting theorems
while filtering out less important statements (see \cite{GaoLiCheng2019,PuzisGaoSutcliffe2006}).

\subsection{Combinatorial explosion and computational complexity}
By using the classes of the equivalence relations, the number of statements
to be checked can be decreased significantly.

For each possible statement, a numerical check is first performed.
Unfortunately, for some exotic coordinates of the points in the figure,
the numerical check can be completely misleading.
For example, very large numbers can sometimes result in numerically unstable computations.
Regardless, if a numerical check is positive, then the statement is added
to the list of conjectures, but if it is negative, no conjecture is registered.
As a consequence, while our implementation
may miss some true statements (due to numerical errors), it will not output false statements.

For each conjecture, a symbolic check will be performed. If the symbolic check is positive,
then the statement will be saved as a theorem. If the symbolic check is negative, then the statement will be
removed from the list of conjectures. If the symbolic check cannot determine whether
a conjecture is true or false, the conjecture is removed from the list.

A special case of a conjecture is $P_1=P_2$ for each two geometric points. If this conjecture cannot be proven or disproven symbolically, then the discovery process will be halted. The user will then be notified
that the construction must be redrawn in a different way, or else no output can be produced. This is required to maintain the consistency of the internal data.

Symbolic checks usually require more time than numerical verifications. The underlying computation
uses Gr\"obner bases that require at most double exponential time of the number of variables \cite{mayrmeyer82} according
to the algebraic translation of the given figure. Usually, the number of variables are double the number of
geometric points in the figure (since there are two coordinates for each).

GeoGebra internally sets 5 seconds for the maximal execution time of each symbolic check. After timeout
the result of the symbolic check will be undecided.

Complexity can be decreased by hiding some points in the figure. This
allows for faster processing and avoids an overly cluttered output
containing too much information.

\subsection{User interface enhancements}

Currently only points can be investigated. In a future version, a set of points, segments,
lines, circles or a set of these could be allowed as input.

At the moment the computation process cannot be interrupted by the user. Given a large number of points
in the figure, the calculation can be time consuming. For example, investigating the relationships
of a regular 20-gon may require about 4 minutes on a modern personal computer (in our test
a Lenovo ThinkPad E480 with an i7 processor, 16 GB RAM, Ubuntu Linux 18.04, was used).

The current version, which is based on GeoGebra Classic 5, performs better than the one on Classic 6. The
latter is a web implementation of the GeoGebra application and uses a WebAssembly compilation
of the computer algebra system Giac. Even if the code is reasonably fast as embedded code
in a web page, Classic 6 underperforms the native technology: the same hardware
is unable to handle the input of the regular 20-gon,
with Google Chrome 83 crashing after 12 minutes of computation.

\subsection{Angles}
In a complex algebraic geometry setting, the study of angles is not as straightforward as investigating other objects.
For a future version, however, this feature would be an important improvement.

By combining algebraic and pure geometric observations, however, simple theorems on angle equality could
be easily detected. For example, Fig.~\ref{contest2} states that points $A$, $F$, $P$, $Q$ are concyclic.
The inscribed angle theorem automatically implies $\angle QAP=\angle QFP$, among others.

\subsection{Stepwise suggestions}

Prior research (see \cite[p.~46]{matematech}) proposed that collecting the interesting new objects
in a figure could be done stepwise, similarly to GeoGebra's former feature ``special objects.''
For our midline theorem example (Fig.~\ref{midline1}), this meant that after constructing the triangle $ABC$ and then midpoint $D$,
the system automatically displayed the segments $BD$ and $CD$.
The user could then accept these newly generated segments or remove them from the system.
Then, by creating midpoint $E$, the system could show lines $AB$ and $DE$ to visualize parallelism.

However, this ``special objects'' feature was recently removed from GeoGebra after feedback that many users found this feature confusing.
How stepwise suggestions might be implemented in GeoGebra Discovery remains a topic of future investigation.

\section{Related work}

We now discuss several projects that share some similarity to GeoGebra Discover but differ 
in significant  ways.

First, GeoGebra Discovery is not the first tool that systematically displays confirmed theorems in a geometric
figure. We refer the reader to 
\begin{itemize}
\item Zlatan Magajna's \textit{OK Geometry}\footnote{\url{www.ok-geometry.com}} \cite{Magajna2011},
\item Jacques Gressier's \textit{G\'eom\'etrix}\footnote{\url{geometrix.free.fr}} and
\item \textit{Java Geometry Expert}\footnote{\url{https://github.com/yezheng1981/Java-Geometry-Expert}}
(JGEX) \cite{Ye_2011}.
\end{itemize}
These systems are available free of charge, but the source code is available only for JGEX.
On the other hand, GeoGebra Discovery focuses on an intuitive user interface
and proofs in the most mathematical sense.

Second, we note that there is a growing interest in creating algorithms related to success completion of 
secondary school or undergraduate mathematics entrance exams. (See \cite{robot-china,text-diagram,japan-todai}, among others.)
Sometimes these projects rely significantly on techniques used in the underlying computational methods.
Also, these projects are often related to artificial intelligence and Big Data rather than to
computational mathematics.

Third, we mention a theoretical issue. The idea to store a geometric point only once if it is identical to another one was previously described in Kortenkamp's work \cite[9.3.1]{ulli99}. This 
concept is a main design element in the dynamic geometry software Cinderella, which
never stores a geometric point twice if the two variants are identical in general.

GeoGebra has a different design concept by allowing the user an arbitrary number of identical points to
be defined. From the theorem prover's point of view, GeoGebra's concept is more difficult to handle,
and a kind of translation is required to have a different data structure by using the
concepts from Section \ref{sec2}.

Finally, we note that GeoGebra Discovery proves the truth in a different manner from Cinderella,
with Cinderella using a probabilistic method,
and GeoGebra Discovery literally \textit{proving} all the deduced facts.

\section{Conclusion}

We described a prototype of the \texttt{Discover} command that is available in an experimental version of GeoGebra, called GeoGebra Discovery. 
This tool facilitates geometric analysis to determine important properties and theorems
based on the user's input.
Our current implementation can be directly
downloaded from \url{https://github.com/kovzol/geogebra/releases/tag/v5.0.641.0-2021Sep03}.

Active testing and sample outputs of the \texttt{Discover} command can be seen at
\url{https://prover-test.geogebra.org/job/GeoGebra_Discovery-discovertest/}.
Currently, 19 test problems are successfully demonstrated, 
including Brahmagupta's theorem,
Napoleon's theorem, Thebault's first two theorems, as well as some simpler statements.
Areas for future development are listed in Section \ref{sec4}.

\section{Acknowledgments}

The \texttt{Discover} command is a result of a long collaboration of several researchers.
The project was initiated by Tom\'as Recio in 2010, and several other researchers joined,
including Francisco Botana and M.~Pilar V\'elez, to name just the most prominent collaborators.
The development and research work was continuously monitored and supported by the GeoGebra Team.
Special thanks to Markus Hohenwarter, project director of GeoGebra.

\bibliography{kovzol,external}

\end{document}